\documentclass[12pt]{article}
\usepackage{amssymb}
\usepackage{amsthm}
\newtheorem{Definition}{Definition}[section]
\newtheorem{Lemma}[Definition]{Lemma}
\newtheorem{Proposition}[Definition]{Proposition}
\newtheorem{Theorem}[Definition]{Theorem}

\newtheorem{Corollary}[Definition]{Corollary}
\title{On orthoposets of numerical events in quantum logic}
\author{Dietmar~Dorninger and Helmut~L\"anger}
\date{}

\begin{document}

\footnotetext[1]{Support of the research of the second author by the Austrian Science Fund (FWF), project 10.55776/PIN5424624, is gratefully acknowledged.}

\maketitle

\begin{abstract}
Let $S$ be a set of states of a physical system and $p(s)$ the probability of the occurrence  of an event when the system is in state $s\in S$. Such a function $p\colon S\to[0,1]$ is known as a numerical event or more accurately an $S$-probability. A set $P$ of numerical events including the constant functions $0$ and $1$ and $1-p$ with every $p\in P$ becomes a poset when ordered by the order of real functions and can serve as a general setting for quantum logics. We call such a poset $P$ a general set of events (GSE). The thoroughly investigated algebras of $S$-probabilities (including Hilbert logics), concrete logics and Boolean algebras can all be represented within this setting. In this paper we study various classes of GSEs, in particular those that are orthoposets and their interrelations and connections to known logics. Moreover, we characterize GSEs as posets by means of states and discuss the situation for GSEs to be lattices.
\end{abstract} 

{\bf AMS Subject Classification:} 06C15, 06E99, 03G12, 81P10

{\bf Keywords:} Quantum logic, numerical event, $S$-probability, state, orthoposet,  orthomodularity
 
\section{Introduction}

A crucial question in physics is whether one deals with a classical physical system or a quantum mechanical one. Among other possibilities, the difference can be described by means of the underlying quantum logic, which in classical physics is a Boolean algebra and in quantum physics a generalization of this concept in the form of special posets. In this paper we will focus on orthoposets of numerical events. Numerical events are defined as follows:

Let $S$ be a set of states of a physical system and $p(s)$ be the probability of the occurrence of an event (pertaining to a certain observable) when the system is in state $s\in S$. The function $p\colon S\to[0,1]$ is known as a {\em numerical event}, {\em multidimensional probability} or, more precisely, an {\em $S$-probability} (cf.\ \cite{BM91} and \cite{BM93}). $S$-probabilities are real functions that can be ordered. Thereby we will always assume that the obtained poset will contain the constant functions $0$ and $1$ for which we will use the same symbols as for the integers $0$ and $1$, respectively. The first class of posets we will specify this way is defined as follows:

\begin{Definition}\label{def1}
{\rm(\cite{BM91} and \cite{BM93})} Let $P$ be a set of $S$-probabilities comprising the constant functions $0$ and $1$ ordered by the partial order $\le$ of real functions and $+$ and $-$ denote the sum and difference of real functions, respectively. $P$ is called an {\rm algebra of $S$-probabilities} or an {\rm algebra of numerical events}, if it satisfies the following axioms:
\begin{enumerate}
\item[{\rm(1)}] $0\in P$;
\item[{\rm(2)}] $p':=1-p\in P$ for all $p\in P$;
\item[{\rm(3)}] if $p,q\in P$ are orthogonal, i.e.\ $p\le q'$, denoted by $p\perp q$, then $p+q\in P$;
\item[{\rm(4)}] if $p,q,r\in P$ and $p\perp q\perp r\perp p$, then $p+q+r\in P$.
\end{enumerate}
\end{Definition}

Obviously, due to $0\perp p\perp q\perp 0$, axiom (3) is a special case of axiom (4), but it will later be important to distinguish between axioms (3) and (4). If $p\perp q$ for $p,q\in P$ then $p+q$ is the supremum $p\vee q$ of $p$ and $q$, which is a consequence of axiom (4) (see \cite{MT}). If the infimum of two $S$-probabilities $p,q$ exists we will denote it by $p\wedge q$.
 
The definition of algebras of $S$-probabilities is motivated by classical event fields, for which the pairwise orthogonality of a triple $A,B,C$ of events implies $A\subseteq B'\cap C'=(B\cup C)'$, which for $S$-probabilities can be translated to $p\le1-(q+r)$.

As already pointed out by M.J.~M\c aczy\'nski and T.~Traczyk \cite{MT} an algebra of numerical events is an orthomodular poset which admits a full set of states and any orthomodular poset that admits a full set of states can is isomorphic to an algebra of numerical events. In particular, all Boolean logics and all Hilbert logics are algebras of $S$-probabilities.(For the question how Hilbert logics are represented as algebras of $S$-probabilities cf.\ e.g.\ \cite{DLM}.)

If axiom (4) is omitted in Definition~\ref{def1}, then the arising poset is known as a GFE ({\em generalized fields of events}) (cf.\ \cite D -- \cite{DL21}). In \cite D it was proved that if elements $p$ and $q$ of some GFE $P$ satisfy $p\le q$ then $q-p\in P$. If for a GFE $P$ $p\perp q$ implies $p+q=p\vee q$ for $p,q\in P$, then $P$ is an algebra of $S$-probabilities (cf.\ \cite{DL18}). Moreover, a GFE such that the values of its $S$-probabilities can only be $0$ or $1$ is already an algebra of $S$-probabilities. In this case the algebra of $S$-probabilities is a {\em concrete logic}, which is a logic that can be represented by sets (cf.\ \cite P).

We further mention that within algebras of $S$-probabilities any element different from $0$ and $1$ that is neither $\le1/2$ nor $\ge1/2$ is called {\em varying} (cf.\ \cite{DDL}). An $S$-probability will be called {\em proper}, if it is varying or equal to $0$ or $1$ (cf.\ \cite{DL18} and \cite{DL21}).

If axioms (3) and (4) are both omitted in Definition~\ref{def1} then we will call the arising poset a {\em general set of {\rm(}numerical{\rm)} events}, in short a GSE.

An important property of algebras of $S$-probabilities is that they are orthoposets which means that besides $p''= p$ ($\,'$ is an involution), $p\le q\Rightarrow q'\le p'$ for $p,q\in P$ ($\,'$ is antitone) also $p\vee p'=1$ and $p\wedge p'=0$ for all $p$ ($\,'$ is a complementation) (cf.\ e.g.\ \cite T).

A GSE $P$ of $S$-probabilities will be called {\em complemented} if $\,'$ is a complementation on the bounded poset $(P,\le,0,1)$. Clearly every GSE is a bounded poset with an antitone involution. Crucial is the following property:

\begin{Proposition}\label{prop1}
Let $P$ be a {\rm GSE}. Then $P$ is complemented and hence an orthoposet if and only if every element of $p$ is proper.	
\end{Proposition}

\begin{proof}
Let $p,q\in P$. First assume $P$ to be complemented. Further, assume $p\ne0,1$. Then $p\le1/2$ would imply $p\le1/2\le p'$ and hence $p=p\wedge p'=0$, a contradiction. Dually, $p\ge1/2$ would imply $p'\le1/2\le p$ and hence $p=p\vee p'=1$, a contradiction. This proves $p$ to be proper. Conversely, assume every element of $P$ to be proper. Then $p,p'\le q$ implies $q\ge1/2$ and hence $q=1$ showing $p\vee p'=1$. Dually, $q\le p,p'$ implies $q\le1/2$ and hence $q=0$ showing $p\wedge p'=0$. Thus $P$ is complemented.
\end{proof}

Since every algebra of $S$-probabilities is a GFE and every GFE is a GSE this means that every algebra of $S$-probabilities is a GFE of proper $S$-probabilities, that is to say a complemented GFE, and those are GSEs of proper $S$-probabilities, i.e.\ complemented GSEs.

In this paper we will characterize complemented GSEs similarly to algebras of numerical events which are up to isomorphism exactly those orthomodular posets that admit a full set of states. We will study various classes of GSEs which have distinguished features of algebras of numerical events and constitute quantum logics of their own accord. In particular, posets of orthogonally composable $S$-probabilities and GSEs that are lattices will be discussed.

\section{A characterization of general sets of numerical events as posets}

\begin{Definition}\label{def2}
Let $\mathbf P=(P,\le,{}',0,1)$ be a bounded poset with an antitone involution and $m\colon P\to[0,1]$. Then $m$ is called a {\em state} on $\mathbf P$ if it satisfies the following conditions:
\begin{enumerate}
\item[{\rm(i)}] $m(0)=0$ and $m(1)=1$,
\item[{\rm(ii)}] $p,q\in P$ and $p\le q$ together imply $m(p)\le m(q)$,
\item[{\rm(iii)}] $m(p')=1-m(p)$ for all $p\in P$.
\end{enumerate}
Let $M$ be a set of states on $\mathbf P$. Then $M$ is called
\begin{enumerate}
\item[{\rm(iv)}] {\em full} if $p,q\in P$ and $m(p)\le m(q)$ for all $m\in M$ together imply $p\le q$,
\item[{\rm(v)}] {\em proper} if $p\in P$, $m_1,m_2\in M$, $m_1(p)\ne0$ and $m_2(p)\ne1$ together imply that there exist $m_3,m_4\in M$ with $m_3(p)<1/2<m_4(p)$.
\end{enumerate}	
\end{Definition}

\begin{Theorem}\label{th1}
Up to isomorphism the general sets of proper $S$-probabilities are exactly the bounded posets with an antitone involution having a full and proper set of states.	
\end{Theorem}

\begin{proof}
First let $P$ be a general set of proper $S$-probabilities. Then $\mathbf P=(P,\le,{}',0,1)$ is a bounded poset with an orthocomplementation. For every $s\in S$ define $m_s\colon P\to[0,1]$ by $m_s(p):=p(s)$ for all $p\in P$ and put $M:=\{m_s\mid s\in S\}$. Let $t\in S$. Then
\begin{enumerate}
\item[(i)] $m_t(0)=0(t)=0$ and $m_t(1)=1(t)=1$,
\item[(ii)] $p,q\in P$ and $p\le q$ together imply $m_t(p)=p(t)\le q(t)=m_t(q)$,
\item[(iii)] $m_t(p')=p'(t)=1-p(t)=1-m_t(p)$ for all $p\in P$.
\end{enumerate}
This shows that $M$ is a set of states on $\mathbf P$. If $p,q\in P$ and $m_s(p)\le m_s(q)$ for all $s\in S$ then $p(s)\le q(s)$ for all $s\in S$, i.e.\ $p\le q$. This proves $M$ to be full. If $p\in P$, $s_1,s_2\in S$, $m_{s_1}(p)\ne0$ and $m_{s_2}(p)\ne1$ then $p(s_1)\ne0$ and $p(s_2)\ne1$ and hence $p\ne0,1$. Since $p$ is proper there exist $s_3,s_4\in S$ with $p(s_3)<1/2<p(s_4)$, i.e.\ $m_{s_3}(p)<1/2<m_{s_4}(p)$ showing $M$ to be proper. Altogether, we have shown that $\mathbf P$ is a bounded poset with an antitone involution having a full and proper set of states.

Conversely, let $\mathbf P=(P,\le,{}',0,1)$ be a bounded poset with an antitone involution having a full and proper set $S$ of states. Define $f\colon P\to[0,1]^S$ by $\big(f(p)\big)(s):=s(p)$ for all $p\in P$ and all $s\in S$. Since $S$ is a full set of states on $\mathbf P$, for every $p,q\in P$, $p\le q$ is equivalent to $f(p)\le f(q)$. This shows that $f$ is an isomorphism from $(P,\le)$ to $\big(f(P),\le\big)$ and hence also an isomorphism from $(P,\le,0,1)$ to $\big(f(P),\le,0,1\big)$. Moreover, we have
\[
\big(f(p')\big)(s)=s(p')=1-s(p)=1-\big(f(p)\big)(s)=\big(f(p)\big)'(s)
\]
for all $p\in P$ and all $s\in S$ and therefore $f(p')=\big(f(p)\big)'$ for all $p\in P$. This secures that $f$ is an isomorphism from $\mathbf P$ to $\big(f(P),\le,{}',0,1\big)$. Since $\mathbf P$ is a bounded poset with an antitone involution, $f(P)$ is a general set of $S$-probabilities. Assume $p\in P\setminus\{0,1\}$. Then there exist $s_1,s_2\in S$ with $p(s_1)\ne0$ and $p(s_2)\ne1$ from which we infer $s_1(p)=\big(f(p)\big)(s_1)\ne0$ and $s_2(p)=\big(f(p)\big)(s_1)\ne1$. Because $S$ is a proper set of states on $\mathbf P$ there exist $s_3,s_4\in S$ with $s_3(p)<1/2<s_4(p)$, i.e.\ $\big(f(p)\big)(s_3)<1/2<\big(f(p)\big)(s_4)$, ensuring that $f(p)$ is proper. Hence any element of $f(P)$ is proper. Altogether we have that $f(P)$ is a general set of proper $S$-probabilities.
\end{proof}
    
\section{Posets of orthogonally composable numerical events}  

\begin{Definition}
A {\em near-generalized field of events {\rm(NGFE)}} is a {\rm GSE} $P$ having the property that every element of $P\setminus\{0\}$ that is not an atom of $(P,\le)$ is the sum of two elements of $P\setminus\{0\}$.
\end{Definition}

\begin{Lemma}\label{lem2}
Let $P$ be a {\rm GFE}. Then $P$ is an {\rm NGFE}.
\end{Lemma}

\begin{proof}
Let $p\in P\setminus\{0\}$ not being an atom of $P$. Then there exists some $q\in P$ with $0<q<p$. Now we have $p=q+(p-q)$ with $q,p-q\in P\setminus\{0\}$.
\end{proof}	

\begin{Definition}
Let $P$ be a {\rm GSE} and $p\in P$. The {\em element} $p$ is called {\em of finite length} if there is no infinite chain between $0$ and $p$. The {\em GSE} $P$ is said to be {\em of finite length} if every of its elements has this property.
\end{Definition}

\begin{Lemma}\label{lem3}
Let $P$ be an {\rm NGFE} and $p\in P\setminus\{0\}$ of finite length. Then $p$ is the sum of finitely many atoms of $P$.
\end{Lemma}

\begin{proof}
If $p$ is an atom of $P$, we are done. Otherwise there exist $p_1,p_2\in P\setminus\{0\}$ with $p_1+p_2=p$. If $p_1$ and $p_2$ are atoms of $P$, we are done. If $p_1$ is not an atom of $P$ then there exist $p_3,p_4\in P\setminus\{0\}$ with $p_3+p_4=p_1$ and we have $p=p_3+p_4+p_2$. If $p_2$ is not an atom of $P$ then there exist $p_5,p_6\in P\setminus\{0\}$ with $p_5+p_6=p_2$ and we have $p=p_1+p_5+p_6$. Since $p$ is of finite length, this procedure must terminate after a finite number of steps.
\end{proof}

\begin{Definition}
For a {\rm GSE} $P$ we define the following property:
\begin{itemize}
\item[{\rm(P)}] $p,q\in P$ and $p\le q$ together imply $q-p\in P$.
\end{itemize}
\end{Definition}

\begin{Lemma}\label{lem4}
Let $P$ be a {\rm GSE}. Then the following are equivalent:
\begin{enumerate}
\item[{\rm(i)}] $P$ is a {\rm GFE},
\item[{\rm(ii)}] $P$ has property {\rm(P)}.		
\end{enumerate}
\end{Lemma}

\begin{proof}
Let $p,q\in P$. \\
(i) $\Rightarrow$ (ii): \\
This was proved in \cite D. \\
(ii) $\Rightarrow$ (i): \\
Assume $p\perp q$, i.e.\ $p\le1-q$. Since $P$ has property (P), we obtain $(1-q)-p\in P$ and hence $p+q=1-\big((1-q)-p\big)\in P$.
\end{proof}

\begin{Theorem}
Let $P$ be an {\rm NGFE} of finite length. Then the following are equivalent:
\begin{enumerate}
\item[{\rm(i)}] $P$ is an algebra of $S$-probabilities,
\item[{\rm(ii)}] any sum of finitely many pairwise orthogonal atoms of $P$ belongs to $P$.
\end{enumerate}
\end{Theorem}

\begin{proof}
$\mbox{}$ \\	
(i) $\Rightarrow$ (ii): \\
As mentioned in the introduction, the sum of two orthogonal elements of $P$ coincides with their supremum. Now let $n$ be an integer $>1$, assume that the sum of $n$ pairwise orthogonal elements of $P$ coincides with their supremum in $P$ and let $a_1,\dots,a_{n+1}$ be pairwise orthogonal elements of $P$. Since $a_i\le a_{n+1}'$ for all $i=1,\dots,n$ we have $a_1+\dots+a_n=a_1\vee\cdots\vee a_n\le a_{n+1}'$ and hence $a_1+\dots+a_{n+1}=(a_1+\dots+a_n)+a_{n+1}=(a_1\vee\cdots\vee a_n)\vee a_{n+1}\in P$. This shows that $P$ is closed under the sum of finitely many orthogonal elements. \\
(ii) $\Rightarrow$ (i): \\
Let $p,q,r\in P$ and assume $p\perp q\perp r\perp p$. If two of $p,q,r$ are $0$ then $p+q+r\in\{p,q,r\}\subseteq P$. Hence assume that at most one of $p,q,r$ equals $0$. Suppose $r=0$. According to Lemma~\ref{lem3} there exist positive integers $k$ and $m$ and atoms $a_1,\dots,a_{k+m}$ of $P$ with
\begin{eqnarray*}
& & p=a_1+\dots+a_k, \\
& & q=a_{k+1}+\dots+a_{k+m}.
\end{eqnarray*}
Since $p+q\le1$, the atoms $a_1,\dots,a_{k+m}$ of $P$ are pairwise orthogonal and hence $p+q+r=p+q=a_1+\dots+a_{k+m}\in P$. The cases $p=0$ and $q=0$ can be treated in a similar way. Now assume $p,q,r\ne0$. Again according to Lemma~\ref{lem3} there exist positive integers $k$, $m$ and $n$ and atoms $a_1,\dots,a_{k+m+n}$ of $P$ with
\begin{eqnarray*}
& & p=a_1+\dots+a_k, \\
& & q=a_{k+1}+\dots+a_{k+m}, \\
& & r=a_{k+m+1}+\dots+a_{k+m+n}.
\end{eqnarray*}
Since $p+q,q+r,r+p\le1$ the atoms $a_1,\dots,a_{k+m+n}$ of $P$ are pairwise orthogonal showing $p+q+r=a_1+\dots+a_{k+m+n}\in P$.
\end{proof}

Though the notion of a concrete logic is mostly used in connection with orthomodular posets which can be represented by sets we will also use this concept for NGFEs which can be represented by subsets of a set.

\begin{Definition}
Let $P$ be a {\rm NGFE} of finite length and $A_p$ for $p\in P\setminus\{0\}$ be a set of atoms of $P$ summing up to $p$ according to Lemma~\ref{lem3}. Then we define for $P$ the following property:
\begin{enumerate}
\item[{\rm(U)}] For every $p,q\in P\setminus\{0\}$, $p\le q$ is equivalent to $A_p\subseteq A_q$ for any $A_p,A_q$.
\end{enumerate}
\end{Definition} 

\begin{Proposition}\label{prop2}
Let $P$ be an {\rm NGFE} of finite length having property {\rm(U)}. Then every $p\in P\setminus\{0\}$ is the unique sum of atoms of $P$ and $P$ is a concrete logic.
\end{Proposition}

\begin{proof}
Assume $p$ is the sum of the elements of $A_p$ as well as the sum of the elements of the set of atoms $\overline A_p$. Then because of $p\le p$ we have $A_p\subseteq\overline A_p$ and $\overline A_p\subseteq A_p$ which means that $A_p$ is unique. Therefore $A_p\subseteq A_1$ for every $p\in P\setminus\{0\}$ which shows that every $p$ can be represented by a subset of $A_1$ in accordance with the definition of a concrete logic.
\end{proof}

\section{Lattices of numerical events}

For a GFE $P$ we define the following property, which originally had been formulated for algebras of $S$-probabilities (cf.\ \cite{DDL} and \cite D).
\begin{enumerate}
\item[(T)] For all $p,q\in P$ there exists a unique $r\in P$ satisfying $r\ge p,q$ and $(r-p)\wedge(r-q)=0$.
\end{enumerate}

\begin{Theorem}\label{th4}
A {\rm GFE} $P$ has property {\rm(T)} if and only if it is a lattice.
\end{Theorem}

Observe that in \cite D this was proved for algebras of $S$-probabilities. In the first part of the following proof we follow these lines.

\begin{proof}[Proof of Theorem~\ref{th4}]
Let $P$ be a GFE and $p,q\in P$. First assume $P$ to have property (T). Let $r$ denote the unique element of $P$ existing according to property (T). Then $r\ge p,q$. Assume $s\in P$ and $s\ge p,q$. Since $p,q\perp s'$ we have $p+s',q+s'\in P$. Because of property (T) there exists some $t\in P$ satisfying
\[
t\ge p+s',q+s'\mbox{ and }\big(t-(p+s')\big)\wedge\big(t-(q+s')\big)=0.
\]
Now $t-s'\ge p,q$ and $\big((t-s')-p\big)\wedge\big((t-s')-q\big)=0$ showing $t-s'=r$ and hence $s=r+t'\ge r$. So we obtain $r=p\vee q$. According to \cite D, $P$ is a lattice. Conversely, assume $P$ to be a lattice. Then $p,q\le p\vee q$. Let $u\in P$ and assume $u\le(p\vee q)-p,(p\vee q)-q$. Then $p,q\le(p\vee q)-u$ and hence $p\vee q\le(p\vee q)-u$ whence $u=0$. This shows $\big((p\vee q)-p\big)\wedge\big((p\vee q)-q\big)=0$. Now assume that $v\in P$ satisfies $p,q\le v$ and $(v-p)\wedge(v-q)=0$. Then $p\vee q\le v$. Suppose $p\vee q<v$. Then $0<v-(p\vee q)\le v-p,v-q$ contradicting $(v-p)\wedge(v-q)=0$. This shows $v=p\vee q$, and hence $p\vee q$ is the unique element $r$ of $P$ satisfying $r\ge p,q$ and $(r-p)\wedge(r-q)=0$. This proves property (T).
\end{proof}

\begin{Corollary}
Let $P$ be an {\rm NGFE} of proper elements and of finite length having properties {\rm(P)}, {\rm(U)} and {\rm(T)}. Then $P$ is an algebra of numerical events that is a concrete logic, i.e.\ an orthomodular lattice that can be represented by sets. If $P$ is finite then $P$ is a Boolean algebra.
\end{Corollary}

\begin{proof}
Referring to property (P), $P$ is a $GFE$ (Lemma~\ref{lem4}), because of property (T) $P$ is a lattice (Theorem~\ref{th4}) and since the elements of $P$ are proper, $P$ is a ortholattice (Proposition~\ref{prop1}). According to \cite D a GFE which is a ortholattice is an algebra of $S$-probabilities if and only if property (T) holds for all $p,q$ with $p\perp q$ which is guaranteed here even for all pairs $p,q\in P$. Hence $P$ is an algebra of $S$-probabilities which due to property (U) is a concrete logic (Proposition~\ref{prop3}). We point out that an algebra of $S$-probabilities is an orthomodular poset (cf.\ \cite{MT}); in case $P$ is finite, $P$ is even a Boolean algebra because by Proposition \ref{prop2} the representation of every $p\in P\setminus\{0\}$ as the sum of atoms is unique and as shown in \cite{DDL} for finite algebras of $S$-probabilities this implies that this algebra of $S$-probabilities is a Boolean algebra.  
\end{proof}  

\begin{Proposition}\label{prop3}
Let $P$ be a {\rm GSE} which is an ortholattice. Then $P$ is an algebra of $S$-probabilities and hence an orthomodular lattice if and only if for all $p,q\in P$, $p\le q$ implies $q-p=q\wedge p'$.
\end{Proposition}

\begin{proof}
Assume $q-p=q\wedge p'$ for $p\le q$. Then $q-p\in P$ because $q\wedge p'\in P$ and therefore, due to Lemma~\ref{lem4}, $P$ is a GFE. Moreover, within this GFE $p\perp r$ for $p,r\in P$ entails $r'-p=r'\wedge p'$ from which we obtain $r+p=(r'-p)'=(r'\wedge p')'=r\vee p$. As already mentioned in the introduction, a GFE having the property that $p\perp r$ implies $p+r=p\vee r$ is an algebra of $S$-probabilities and hence an orthomodular lattice. Conversely, every algebra of $S$-probabilities which is a lattice is a GSE which is an ortholattice having the property that $p\le q$ implies $q-p=q\wedge p'$ (cf.\ \cite{DDL}).
\end{proof}

Next we deal with the atoms of a GFE which particularly in case of finite structures can be more conveniently used for algorithmic purposes.

\begin{Lemma}\label{lem5}
Let $P$ be a {\rm GFE} which is an orthomodular lattice and $\{a_1,\dots,a_n\}$ a set of pairwise orthogonal atoms of $P$. If, step by step from $m=2$ to $m=n$ {\rm(}which will ensure that the considered sums of atoms exist in $P${\rm)}, $(a_1+\dots+ a_m)\wedge(a_1+\dots+a_{m-1})'=a_m$, then $a_1+\dots+a_m=a_1\vee\cdots\vee a_m$.
\end{Lemma}

\begin{proof}
We use induction on $n$. For $n=2$, $a_1+a_2\in P$ because $P$ is a GFE, and from $a_1\le a_1+a_2$, the orthomodularity of $P$ and $(a_1+a_2)\wedge a_1'=a_2$ we infer that $a_1+a_2=a_1\vee\big((a_1+a_2)\wedge a_1'\big)=a_1\vee a_2$. Assuming $b:= a_1+\dots+a_{n-1}=a_1\vee\cdots\vee a_{n-1}$ then $b+a_n\in P$ because $a_1\vee\cdots\vee a_{n-1}\le a_n'$, and in the same way as for $n=2$ we obtain $a_1+\dots+a_n=b+a_n=b\vee\big((b+a_n)\wedge b'\big)=b\vee a_n=a_1\vee\cdots\vee a_n$.
\end{proof}

For a set $\{a_1,\dots,a_n\}$ of pairwise orthogonal atoms of a GFE which is an orthomodular lattice we define the following property:
\begin{enumerate}
\item[(SJ)] $(a_1+\dots+a_m)\wedge(a_1+\dots+a_{m-1})'=a_m$ for $m=2,\dots,n$,
\end{enumerate}
which according to Lemma~\ref{lem5} ensures that $a_1+\dots+a_m=a_1\vee\cdots\vee a_m$ for $m=2,\dots,n$.

We will say that a GFE $P$ which is an orthomodular lattice has the property $(SJ)$ if any set of pairwise orthgogonal elements of $P$ has property $(SJ)$. 

\begin{Theorem}
Let $P$ be {\rm GFE} of proper elements and of finite length which is an orthomodular lattice. Then $P$ is an algebra of $S$-probabilities if and only if it has property {\rm(SJ)}.
\end{Theorem}

\begin{proof}
First assume $P$ to be an algebra of $S$-probabilities, let $b_1,\dots,b_n$ be pairwise orthogonal elements of $P$ and assume $2\le m\le n$. It is well-known and easy to see that $b_1+b_2=b_1\vee b_2$ and, by induction on $n$, that $b_1+\dots+b_n=b_1\vee\cdots\vee b_n$. Now we have
\begin{eqnarray*}
& & (b_1\vee\cdots\vee b_{m-1})+\big((b_1\vee\cdots\vee b_m)\wedge(b_1\vee\cdots\vee b_{m-1})'\big)= \\
& & =(b_1\vee\cdots\vee b_{m-1})\vee\big((b_1\vee\cdots\vee b_m)\wedge(b_1\vee\cdots\vee b_{m-1})'\big)=b_1\vee\cdots\vee b_m= \\
& & =(b_1\vee\cdots\vee b_{m-1})+a_m
\end{eqnarray*}
which proves property (SJ). Conversely, assume $P$ to have property (SJ). Referring to Lemma~\ref{lem3}, every element of $P$ is the sum of atoms of $P$. These atoms are pairwise orthogonal because the sum of every two of them is $\le1$. If $p\perp q$ for $p,q\in P$ with the representation as sums of atoms as
\begin{eqnarray*}
& & p=a_1+\dots+a_m, \\
& & q=a_{m+1}+\dots+a_{m+t},
\end{eqnarray*}
then $p+q=a_1+\dots+a_m+a_{m+1}+\dots+a_{m+t}$ where all $a_i$ are pairwise different (because $2a\le1$ for a common atom $a$ cannot occur) and pairwise orthogonal. Therefore $p+q=p\vee q$, from which we can infer that $P$ is an algebra of $S$-probabilities, as already stated in the introduction.
\end{proof}

{}

Authors' addresses:

Dietmar Dorninger \\
TU Wien \\
Faculty of Mathematics and Geoinformation \\
Institute of Discrete Mathematics and Geometry \\
Wiedner Hauptstra\ss e 8-10 \\
1040 Vienna \\
Austria

Helmut L\"anger \\
TU Wien \\
Faculty of Mathematics and Geoinformation \\
Institute of Discrete Mathematics and Geometry \\
Wiedner Hauptstra\ss e 8-10 \\
1040 Vienna \\
Austria, and \\
Palack\'y University Olomouc \\
Faculty of Science \\
Department of Algebra and Geometry \\
17.\ listopadu 12 \\
771 46 Olomouc \\
Czech Republic \\
helmut.laenger@tuwien.ac.at

\end{document}